\documentclass[10pt,conference]{IEEEtran}
\IEEEoverridecommandlockouts
\usepackage{cite}
\usepackage{amsmath,amssymb,amsfonts,fancyhdr,amsthm}
\usepackage{graphicx}
\usepackage{textcomp}
\usepackage{xcolor}
\def\BibTeX{{\rm B\kern-.05em{\sc i\kern-.025em b}\kern-.08em
    T\kern-.1667em\lower.7ex\hbox{E}\kern-.125emX}}
\usepackage{subfigure}
\usepackage[ruled,vlined,linesnumbered]{algorithm2e}
\usepackage{algpseudocode}
\usepackage{hyperref}

\usepackage{newtxtext,newtxmath}

\newtheorem{theorem}{Theorem}
\newtheorem{lemma}{Lemma}

\newtheorem{remark}{Remark}

\DeclareMathOperator*{\argmin}{arg\,min}
\begin{document}

\title{Low-Complexity Blind Estimator of SNR and MSE for mmWave Multi-Antenna Communications\\
\thanks{This work was supported by Institute of Information \& communications Technology Planning \& Evaluation (IITP) grant funded by the Korea government (MSIT) (No. RS-2024-00442085, No. RS-2024-00398157).}
}

\author{\IEEEauthorblockN{Hanyoung Park and Ji-Woong Choi}
\IEEEauthorblockA{
\textit{Department of Electrical Engineering and Computer Science,
DGIST,
Daegu, South Korea} \\
e-mail: \{prkhnyng, jwchoi\}@dgist.ac.kr}
}

\maketitle

\begin{abstract} 
To enhance the robustness and resilience of wireless communication and meet performance requirements, various environment-reflecting metrics, such as the signal-to-noise ratio (SNR), are utilized as the system parameter. To obtain these metrics, training signals such as pilot sequences are generally employed. However, the rapid fluctuations of the millimeter-wave (mmWave) propagation channel often degrade the accuracy of such estimations. To address this challenge, various blind estimators that operate without pilot have been considered as potential solutions. However, these algorithms often involve a training phase for machine learning or a large number of iterations, which implies prohibitive computational complexity, making them difficult to employ for real-time services and the system less resilient to dynamic environment variation. In this paper, we propose blind estimators for average noise power, signal power, SNR, and mean-square error (MSE) that do not require knowledge of the ground-truth signal or involve high computational complexity. The proposed algorithm leverages the inherent sparsity of mmWave channel in beamspace domain, which makes the signal and noise power components more distinguishable.
\end{abstract}

\begin{IEEEkeywords}
Resilience, blind estimation, SNR, MSE, mmWave, sparsity, beamspace.
\end{IEEEkeywords}

\section{Introduction}
In wireless communication systems, various parameters reflecting the propagation environment are utilized to make the system resilient to the dynamic variation of the channel environment and signal distortions. 
For example, noise power and signal-to-noise ratio (SNR) are required for various applications such as 
precoding\cite{ref:precoding}, channel estimation\cite{ref:channelest}, 
equalization\cite{ref:equalization}, and scheduling\cite{ref:scheduling}. 
These applications help achieve higher quality-of-service (QoS) and improvements in performance, such as error rate, data rate, energy efficiency, and latency.
Most of these frameworks are based on parameter estimation using pilot sequences, which facilitates accurate system estimation.

Meanwhile, with the advances in wireless communication technology, abundant bandwidth is required to support broader services and meet increasing data rate demands.
Therefore, millimeter-wave (mmWave) and terahertz (THz) systems have been considered for the solution of this demand. 
Due to the properties of high-frequency electromagnetic waves, it has shorter channel coherence time in the mmWave and THz frequency bands. 
Moreover, due to the nature of poor scattering and strong path loss, it is highly sensitive to the change of blockage\cite{ref:mmWavesurvey}.
To make the system resilient to this faster fluctuation, more frequent transmission of pilot sequence to perform estimation more frequently can be the intuitive solution.
However, this approach may lead to the pilot overhead increment, and it implies the inefficiency in terms of radio resource and reduction in data rate.

To meet this demand, blind estimators can be considered, which have been proposed in various prior research. 
The works done in \cite{ref:OFDM1, ref:OFDM2, ref:OFDM3} used the inherent characteristics of orthogonal frequency division multiplexing (OFDM), including the cyclic prefix and the periodic property of the OFDM waveform in modulation-specific manners. These estimators show sufficient estimation accuracy, but they require background information about the modulation scheme or have availability only for the OFDM waveform. 
The works done in \cite{ref:EM1} and \cite{ref:EM2} proposed blind estimation of noise power and SNR based on the expectation-maximization (EM) algorithm. Covariance matrices of the received signal in multi-input multi-output (MIMO) systems are also utilized for blind SNR estimation\cite{ref:covariance1, ref:covariance2}. Also, deep learning-based estimators have been proposed in previous works \cite{ref:DL1, ref:DL2, ref:DL3}. However, these algorithms require a large number of iterations, a training phase for machine learning, numerous samples, or excessive computational burden. 
These requirements make the corresponding algorithms unsuitable for real-time operation, particularly in resilient multi-antenna mmWave and THz systems where low-latency and robustness are critical.
Also, the increment in the number of antennas makes the costs of these algorithms more prohibitive.
To this end, a low-complexity estimator has been proposed, but it suffers from limited estimation accuracy\cite{ref:lowblind}.

Therefore, in this paper, we propose low-complexity, computationally-efficient blind estimators of average noise power and SNR with improved estimation accuracy by utilizing the signal sparsity of mmWave multi-antenna propagation environments in beamspace domain. 
The sparsity of signal in mmWave channel environment due to the nature of poor scattering makes the power from the signal concentrated in a few beams, and it makes separation easier to distinguish the power of signal and noise. 
Therefore, we exploit this characteristic to distinguish the signal and noise and estimate their power using the element-wise gain differences of the signal in beamspace domain. 
In addition, we propose a blind mean-square error (MSE) estimator which can be utilized for beamspace denoiser without knowledge of ground-truth signal. 
For MSE evaluation in blind scenarios, Stein's unbiased risk estimate (SURE) can be utilized, but it requires knowledge of noise variance\cite{ref:sure}. 
Accordingly, we devise blind MSE estimator by integrating SURE and the noise power estimator proposed in this paper.
The proposed algorithms demonstrate improved estimation accuracy with a single snapshot of the signal captured by the multi-antenna receiver without relying on iterative optimization or deep learning, enabling the system to track its parameters more quickly and enhance its resilience against the harsh mmWave propagation environment.

\subsubsection*{Notation} Uppercase and lowercase bold symbols denote matrices and column vectors, respectively. For a matrix $\mathbf{A}$, $\mathbf{a}_i$ indicates its $i$th column vector,
$a_j$ represents its $j$th element.
$(\cdot)^H$ presents Hermitian transpose, and $\mathbf{I}_N$ is $N\times N$ identity matrix.
$\|\cdot\|$ represents the $\ell_2$ norm and $\|\cdot\|_1$ is the $\ell_1$ norm.
\section{System Model}
In our model, we consider a mmWave single-input multiple-output (SIMO) uplink system. 
The base station (BS) has $M$ receiver antennas, and a user equipment (UE) has a single antenna.
The received signal vector is determined as
\begin{equation}\label{eqn:systemmodel}
    \mathbf{y}=\mathbf{hs}+\mathbf{n}=\mathbf{x}+\mathbf{n},
\end{equation}
where $\mathbf{h}\in\mathbb{C}^{M\times 1}$ is the channel vector, $\mathbf{s}\in\mathbb{C}$ is the transmit signal of the UE, $\mathbf{n}\in\mathbb{C}^M$ is the i.i.d. zero-mean additive white Gaussian noise (AWGN) vector with variance $N_0$, and $\mathbf{x}\in\mathbb{C}^M$ is the noiseless signal after the propagation through the channel.
We assume that the channel vector is modeled as\cite{ref:poisson}
\begin{equation}
    \mathbf{h}=\sum_{\ell=1}^{L} g_{\ell}\mathbf{a}(\phi_{\ell}),
\end{equation}
where $L$ is the number of propagation paths including a potential line-of-sight (LoS) element, $g_{\ell}$ is the complex channel gain, $\mathbf{a}(\cdot)\in\mathbb{C}^M$ is the steering vector, and $\phi_{\ell}$ is the spatial frequency of the $\ell$-th propagation path of the UE\cite{ref:channel}.

\subsection{Average Noise Power Estimator}
The average noise power of the system is generally defined as
\begin{equation}
    N_0\triangleq \frac{1}{M} \mathbb{E}[\|\mathbf{n}\|^2].
\end{equation}

Here, we propose a low-complexity blind estimator of the average noise power, which is derived by Algorithm \ref{alg:alg1}.
First, to leverage sparsity in the beamspace domain, the received signal is transformed from the antenna domain to the beamspace by discrete Fourier transform (DFT), which is presented as
\begin{equation}
    \mathbf{y}_\mathsf{b}=\mathbf{Fy},
\end{equation}
where $\mathbf{F}$ is the normalized DFT matrix so that the power calculation in the beamspace domain is equivalent to that in the antenna domain without additional scaling, i.e. $\|\mathbf{y}\|^2=\|\mathbf{y}_\mathsf{b}\|^2$.
Subsequently, its element-wise power sequence $\mathbf{p}_{\mathbf{y}_\mathsf{b}}=|\mathbf{y}_\mathsf{b}|^2$ is sorted in ascending order.
Since the propagation channel in the beamspace domain, $\mathbf{h}_\mathsf{b}=\mathbf{Fh}$, is sparse due to the nature of poor scattering, it leads to the sparsity of the noiseless signal $\mathbf{x}_\mathsf{b}=\mathbf{Fx}$ in the beamspace representation.
Accordingly, its noisy observation $\mathbf{y}_\mathsf{b}$ shows a similar property\cite{ref:mmWavesurvey}.
Therefore, the power of the elements corresponding to the signal is significantly larger compared to that of the elements corresponding to the noise.
Furthermore, the signal power is predominantly concentrated in a limited portion in the beamspace, while the other shows near-zero values.
Fig.~\ref{fig:beamspace} shows the power of $\mathbf{p}_{\mathbf{y}_\mathsf{b}}$ and its sorted version $\mathbf{p}_{\mathbf{y}_\mathsf{b}}^\textnormal{sorted}$ in the beamspace domain. 
It shows that the concentrated power makes it easier to distinguish whether power elements originate from the signal or noise.
To exploit this property, we use the finite difference of the sorted power sequence to separate noise components from the noisy observation of signal based on power difference, which is determined as
\begin{equation}
    \Delta\mathbf{p}^\text{sorted}_{\mathbf{y}_\mathsf{b},(i)}=\mathbf{p}^\text{sorted}_{\mathbf{y}_\mathsf{b},(i+1)}-\mathbf{p}^\text{sorted}_{\mathbf{y}_\mathsf{b},(i)},
\end{equation}
where $\mathbf{p}^\text{sorted}_{\mathbf{y}_\mathsf{b},(i)}$ is the $i$th element of the sorted power sequence.
\begin{figure}
    \centering
    \subfigure[]{\includegraphics[width=0.9\columnwidth]{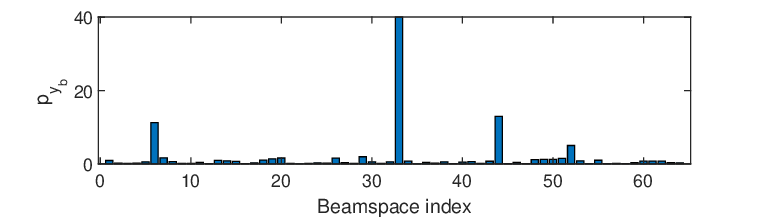}}
    
    \subfigure[]{\includegraphics[width=0.9\columnwidth]{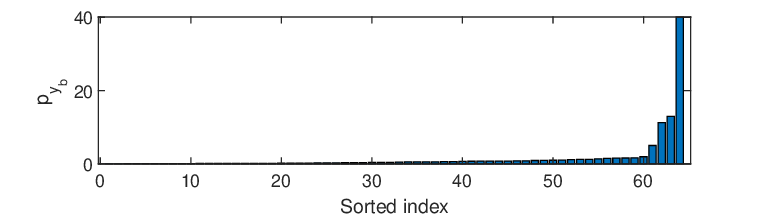}}
    \caption{Power of the (a) noisy signal in the beamspace domain and (b) its sorted version, for $M=64$ and $\text{SNR}=3$ dB.}
    \label{fig:beamspace}
\end{figure}
By the iteration, we find the first index that shows a larger finite difference which satisfies 
\begin{eqnarray}\label{eqn:threshold}
    \Delta\mathbf{p}^\text{sorted}_{\mathbf{y}_\mathsf{b},(i^*)}>\frac{\gamma}{i^*} \sum_{i=1}^{i^*} \mathbf{p}^\text{sorted}_{\mathbf{y}_\mathsf{b},(i)},
\end{eqnarray}
where $\gamma$ is the threshold parameter.
During the iteration, a temporary variable is utilized to store the temporal average value of the power, to integrate the loops into a single loop for reduction of the additional loop iteration.
Note that the elements corresponding to the indices from the first to the $i^*$th are classified as the noise power.
If there exists no value that satisfies \eqref{eqn:threshold}, this indicates that all power components are regarded as noise power. Therefore, we assign $i^*=M$ for this case.
After iteration for finding the index $i^*$, now finally the estimated noise power is calculated as
\begin{equation}
    \widehat{N}_0 = \frac{1}{i^*}\sum_{i=1}^{i^*} \mathbf{p}^\text{sorted}_{\mathbf{y}_\mathsf{b},(i)}.
\end{equation}
The complexity of this noise estimation algorithm is derived as follows. The algorithm requires \textit{(i) DFT}, \textit{(ii) sorting}, \textit{(iii) calculation of finite difference}, \textit{(iv) iteration to find }$i^*$, and \textit{(v) calculation of the mean power of the elements distinguished as noise}. 
Here, calculating the temporary mean power for each iterative step may lead to a prohibitive computational complexity. 
Accordingly, to reduce the computational complexity, we integrated \textit{(iii)--(v)} in one iteration, which requires $\mathcal{O}(M)$.
By employing fast Fourier transform (FFT), \textit{(i)} is conducted with complexity $\mathcal{O}(M\log M)$. Also, \textit{(ii)} can be performed by Quicksort with complexity $\mathcal{O}(M\log M)$. 
Consequently, the computational complexity of the total estimation process in big-O notation is $\mathcal{O}(M\log M)$.
This estimator calculates the average power of the elements which are classified as the noise elements in beamspace domain, without requiring the information of the ground-truth signal or pilot sequence.

\begin{algorithm}[b]
\caption{Blind Average Noise Power Estimator}\label{alg:alg1}
\SetAlgoLined
\textbf{input:} $\mathbf{p}^\textnormal{sorted}_{\mathbf{y}_\mathsf{b}}$

\textbf{initialize}
$i^*=M$

\For{$i=1,...,M-1$}{ 
    Calculate temporary average value of power
    \begin{equation}
        \texttt{avg\_temp} = (\texttt{avg\_temp} \times (i-1) + \mathbf{p}^\textnormal{sorted}_{{\mathbf{y}}_{\mathsf{b},(i)}})\times\frac{1}{i}.\nonumber
    \end{equation} 

    Calculate the finite difference of index $i$ 
    \begin{eqnarray}
        \Delta\mathbf{p}^\textnormal{sorted}_{{\mathbf{y}}_{\mathsf{b},(i)}} = \mathbf{p}^\textnormal{sorted}_{{\mathbf{y}}_{\mathsf{b},(i+1)}} - \mathbf{p}^\textnormal{sorted}_{{\mathbf{y}}_{\mathsf{b},(i)}}.\nonumber
    \end{eqnarray}
    
    \If{$\Delta\mathbf{p}^\textnormal{sorted}_{{\mathbf{y}}_{\mathsf{b},(i)}} \geq \gamma\times\textnormal{\texttt{avg\_temp}}$}{
        $i^*=i$

        \textbf{break}
    }
}

{$\widehat{N}_0=\texttt{avg\_temp}$}

\Return {The estimated noise power $\widehat{N}_0$}    

\end{algorithm}

\subsection{Average Signal Power, SNR, and MSE Estimator}
%
The average signal power of the system is defined as
\begin{equation}
    P_\mathbf{x}\triangleq\frac{1}{M}\mathbb{E}[\|\mathbf{x}\|^2].
\end{equation}
Also, based on this definition, the SNR is defined as
\begin{equation}
    \rho\triangleq \frac{P_\mathbf{x}}{N_0}= \frac{\mathbb{E}[\|\mathbf{x}\|^2]}{\mathbb{E}[\|\mathbf{n}\|^2]}.
\end{equation}
For the estimation of the average signal power, we propose a blind estimator as follows \cite{ref:lowblind},
\begin{equation}
    \widehat{P}_\mathbf{x} \triangleq \frac{1}{M}\max\left\{ \|\mathbf{y}\|^2 - M\widehat{N}_0 ,0 \right\}.
\end{equation}
Moreover, since SNR is defined as the ratio between the average signal power and the average noise power, the SNR estimator can then be computed using the estimated power values, $\widehat{P}_\mathbf{x}$ and $\widehat{N}_0$, as follows,
\begin{equation}
    \widehat{\rho} \triangleq \frac{\widehat{P}_\mathbf{x}}{\widehat{N}_0}.
\end{equation}

Meanwhile, we consider a denoising problem of the noisy signal $\mathbf{y}$ and its general accuracy metric, MSE.
For the denoiser function $\hat{\mathbf{x}}(\mathbf{y})$ which denoises the noisy observation $\mathbf{y}$, the MSE of the denoiser is
\begin{equation}
    \varepsilon^2\triangleq \frac{1}{M}\mathbb{E}[\|\hat{\mathbf{x}}(\mathbf{y}) - \mathbf{x}\|^2].
\end{equation}
Many denoising algorithms utilize MSE to find the optimal denoiser by solving the following optimization problem, which is presented as
\begin{equation}
    \hat{\mathbf{x}}^\star(\mathbf{y})=\argmin_{\hat{\mathbf{x}}(\mathbf{y})} \varepsilon^2.
\end{equation}
However, since this metric requires the knowledge of ground-truth signal $\mathbf{x}$, it cannot be utilized for the blind scenarios.
Fortunately, SURE, an unbiased estimate of the MSE, can be utilized in the blind scenarios under the Gaussian noise conditions\cite{ref:sure}.
Theorem \ref{thm:thm1} shows the unbiased estimate of can be derived by the SURE.
\begin{theorem}[SURE]\label{thm:thm1}
    Consider the signal $\mathbf{x}$ and its noisy observation $\mathbf{y}\sim\mathcal{CN}(\mathbf{x}, N_0\mathbf{I}_M)$. Let $\hat{\mathbf{x}}(\mathbf{y})$ be the denoiser function which estimates $\mathbf{x}$ from $\mathbf{y}$ with weak differentiability. Then,
    \begin{align}
        \mathcal{S} & = \frac{1}{M}\|\hat{\mathbf{x}}(\mathbf{y})-\mathbf{y}\|^2+{N}_0 \nonumber\\ 
    & +\frac{{N}_0}{M}\sum_{m=1}^M \left( \frac{\partial\textnormal{Re}\{\hat{x}_m(y_m)\}}{\partial\textnormal{Re}\{y_m\}} + \frac{\partial\textnormal{Im}\{\hat{x}_m(y_m)\}}{\partial\textnormal{Im}\{y_m\}}-2 \right)
    \end{align}
    is an unbiased estimate of MSE. Thus, $\mathbb{E}[\mathcal{S}]=\varepsilon^2$.
\end{theorem}
\begin{proof}
    Please refer to Appendix \ref{app:1}.
\end{proof}
However, SURE requires knowledge of the noise power $N_0$. 
Accordingly, we propose a blind estimator utilizing SURE and the noise power proposed above, which is determined as
\begin{align}
    \widehat{\varepsilon^2}&\triangleq \frac{1}{M}\|\hat{\mathbf{x}}(\mathbf{y})-\mathbf{y}\|^2+\widehat{N}_0 \nonumber\\ 
    & +\frac{\widehat{N}_0}{M}\sum_{m=1}^M \left( \frac{\partial\textnormal{Re}\{\hat{x}_m(y_m)\}}{\partial\textnormal{Re}\{y_m\}} + \frac{\partial\textnormal{Im}\{\hat{x}_m(y_m)\}}{\partial\textnormal{Im}\{y_m\}}-2 \right).
\end{align}
\begin{remark}
    The proposed estimators do not require knowledge of $\mathbf{x}$ or $N_0$. 
    Accordingly, these estimators can be operated in the blind scenarios. 
    Alternatively, they require the estimated noise power $\widehat{N}_0$. It means that these estimators can be considered as the extended applications of the noise power estimator. Moreover, their estimation performance is influenced by that of the noise power estimator.
\end{remark}
\begin{remark}
    The proposed estimators for the average signal power, SNR, and MSE do not require a training phase for machine learning or iterative optimization, while only requiring low computational complexity of the proposed noise power estimate. Furthermore, they do not require multiple snapshots of received signal, which reduces the complexity.
\end{remark}
\begin{remark}
    The system model described in~\eqref{eqn:systemmodel} assumes a narrowband channel, which implies a flat-fading channel, but the proposed estimators can also be applied for wideband channels with frequency selectivity through subcarrier-wise estimation.
\end{remark}
\section{Simulation Results}
\subsection{Simulation Configuration}
In this section, numerical simulation results are analyzed to verify the performance of the proposed estimators. We consider the 64-element uniform linear array (ULA) at the BS, and its antenna spacing is half-wavelength.
For channel model, we employed QuaDRiGa mmMAGIC UMi\cite{ref:quadriga} with a carrier frequency of 50 GHz.
Without loss of generality, the ground-truth average noise power is assumed to be $N_0=1$ for all simulations.
Also, we performed the simulation with 10000 trials of Monte-Carlo runs.

To evaluate the performance of the MSE estimator, we consider the least absolute shrinkage and selection operator (LASSO) for denoiser, which is broadly used for sparse signal denoising and defined as 
\begin{equation}
    \hat{\mathbf{x}}^\star(\mathbf{y})=\argmin_{\hat{\mathbf{x}}\in\mathbb{C}^M} \frac{1}{2} \|\mathbf{y}-\mathbf{\hat{x}}\|^2+\lambda\|\hat{\mathbf{x}}\|_1,
\end{equation}
where $\lambda>0$ is the denoising threshold parameter.
The denoised result is the solution of the optimization problem, and its closed-form is derived in element-wise manner as follows \cite{ref:lasso}:
\begin{align}
    \hat{x}_m^\star(y_m)= \begin{cases}
        \frac{|y_m|-\lambda}{|y_m|}y_m & \textnormal{if } |y_m| > \lambda \\
        0 & \textnormal{if }|y_m|\leq\lambda,
    \end{cases}
\end{align}
and the optimal denoiser is derived by finding $\lambda$ which makes the $\varepsilon^2$ minimum.

For comparison with the previous work, we consider a single-snapshot-based blind estimator devised in \cite{ref:lowblind} which estimates the noise power based on statistical characteristics with the low computational complexity of $\mathcal{O}(M\log M)$ including FFT, without a training phase for machine learning or iteration.
In addition, we utilize the ground-truth values of the parameters to evaluate the performance of the proposed estimators.

\begin{figure}
    \centering
    \includegraphics[width=0.95\columnwidth]{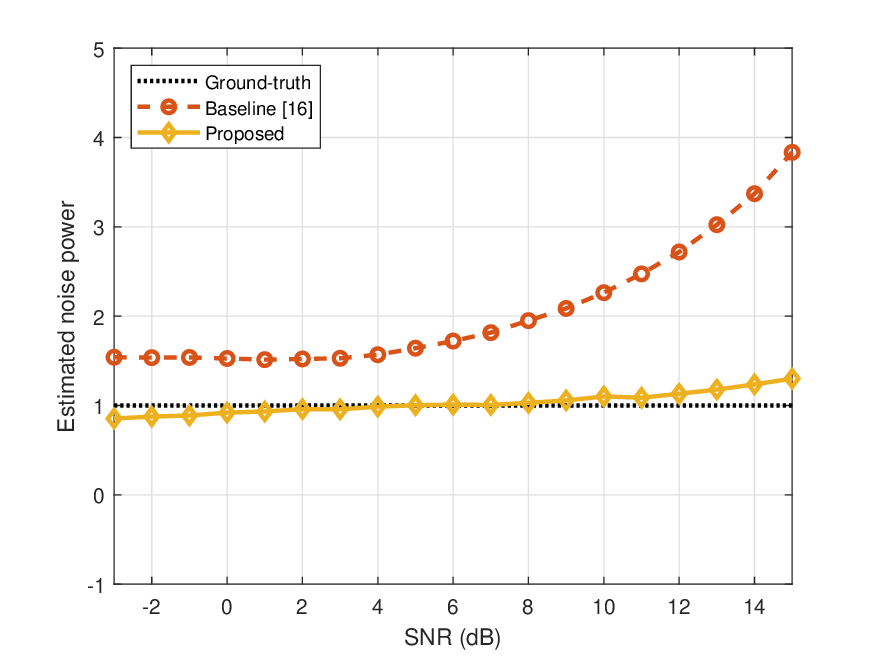}
    \caption{Estimated average noise power depending on SNR.}
    \label{fig:noise}
\end{figure}

\begin{figure}[]
    \centering
    \includegraphics[width=0.95\columnwidth]{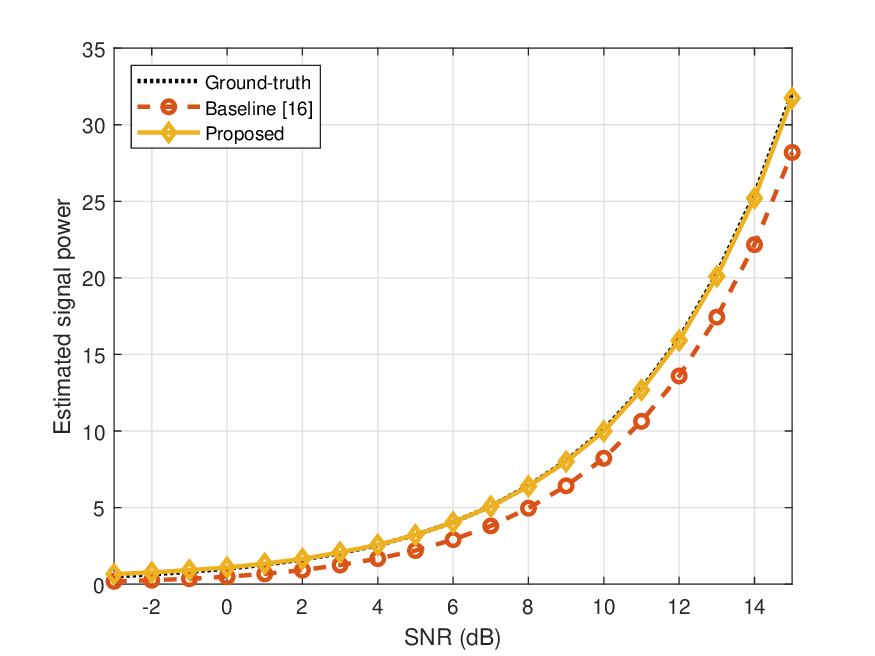}
    \caption{Estimated average signal power depending on SNR.}
    \label{fig:signal}
\end{figure}
{\begin{figure}[t]
    \centering
    \includegraphics[width=0.95\columnwidth]{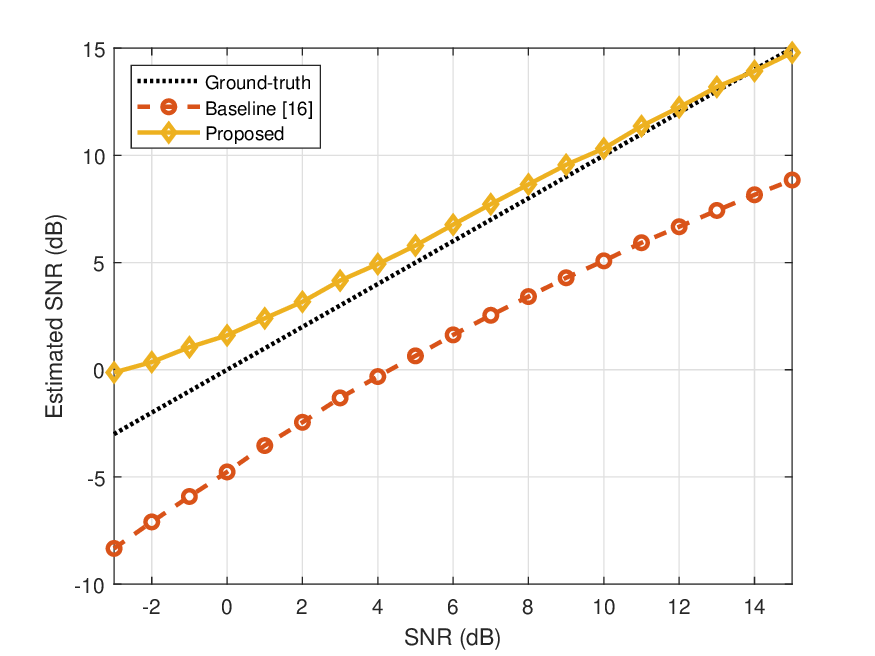}
    \caption{Estimated SNR depending on SNR.}
    \label{fig:snr}
\end{figure}
\begin{figure}[t]
    \centering
    \includegraphics[width=0.95\columnwidth]{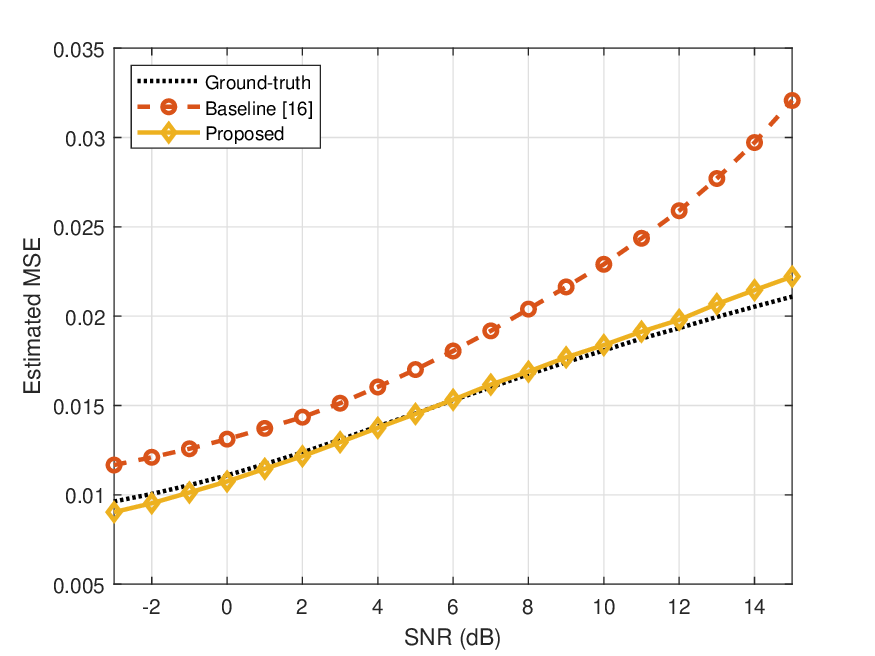}
    \caption{Estimated MSE depending on SNR.}
    \label{fig:mse}
\end{figure}}
\subsection{Performance Analysis}
Fig.~\ref{fig:noise} shows the estimated average noise power depending on SNR.
In lower SNR conditions, the proposed algorithm slightly underestimates the noise power. On the other hand, in higher SNR conditions, it slightly overestimates the noise power.
This is because the probability that noise components with large magnitudes are interpreted as signal increases under low-SNR conditions, while the probability that beam components with small magnitudes are treated as noise increases under high-SNR conditions.
Also, it shows better estimation accuracy than the baseline estimator because it estimates the noise power after separating signal and noise.
Meanwhile, Fig.~\ref{fig:signal} shows that the proposed signal power estimator provides nearly accurate results for all SNR conditions, while achieving improved estimation performance compared to the baseline algorithm.

Fig.~\ref{fig:snr} illustrates the estimated SNR depending on ground-truth SNR.
In higher SNR conditions, the proposed estimator slightly underestimates the SNR.
On the other hand, the proposed estimator overestimates the SNR in low-SNR cases.
It is because the estimator is based on the noise and signal power estimators, and the noise power estimator overestimates and underestimates the noise power in higher and lower SNR conditions, respectively.
Since the proposed estimator showed improved performance compared to the baseline algorithm in noise and signal power estimations, it also shows superior performance in SNR estimation as a result.

Fig.~\ref{fig:mse} shows the estimated MSE depending on SNR when $\lambda$ is 3. 
The proposed estimator slightly underestimates the MSE in low-SNR conditions and overestimates in high-SNR conditions. 
It is because this estimator is also based on the noise power estimator, and the estimation tendency of MSE follows that of the noise power estimates.
Moreover, the proposed algorithm shows better performance than the baseline estimator due to its better accuracy in noise power estimation.

\section{Conclusion}
This paper has proposed low-complexity estimator algorithms for average noise power, signal power, SNR, and MSE, which are critical for resilient communication. These algorithms can be operated in blind scenarios without the knowledge of the ground-truth signal.
Our work utilizes the sparsity of the received signal vector in beamspace domain due to the nature of poor scattering in mmWave systems, which makes the source of the power more distinguishable between signal and noise.
The simulation results have shown that the proposed algorithms can improve estimation performance compared to the baseline algorithms and maintain resilience with low computational complexity.
For future work, we will analyze the influence of estimation on applications and the performance of the proposed estimators in multi-user MIMO systems. In addition, in-depth research on theoretical analysis and bounds is part of our ongoing work.

\appendices
\section{Proof of Theorem \ref{thm:thm1}}\label{app:1}
The purpose of the proof is to show that SURE is an unbiased estimate of MSE $\varepsilon^2$. 
We begin by rewriting the MSE as follows: 
\begin{eqnarray}\label{eqn:msederived}
    \varepsilon^2& = & \frac{1}{M}\mathbb{E}[\|\hat{\mathbf{x}}(\mathbf{y}) - \mathbf{x}\|^2]\nonumber \\
    & = &\frac{1}{M}\mathbb{E}[\|\hat{\mathbf{x}}(\mathbf{y}) - \mathbf{y} - \mathbf{x} + \mathbf{y}\|^2]\nonumber \\
    & = &\frac{1}{M}\mathbb{E}[\|\hat{\mathbf{x}}(\mathbf{y}) - \mathbf{y}\|^2]+ \frac{1}{M}\mathbb{E}[\|\mathbf{y}-\mathbf{x}\|^2]\nonumber \\
    && + \frac{2}{M}\mathbb{E}[(\hat{\mathbf{x}}(\mathbf{y}) - \mathbf{y})^H(\mathbf{x}-\mathbf{y})]\nonumber \\
    & = &\frac{1}{M}\mathbb{E}[\|\hat{\mathbf{x}}(\mathbf{y}) - \mathbf{y}\|^2] + N_0\nonumber \\
    && + \frac{2}{M}\mathbb{E}[(\hat{\mathbf{x}}(\mathbf{y}) - \mathbf{y})^H(\mathbf{x}-\mathbf{y})], 
\end{eqnarray}
since $\mathbb{E}[\|\mathbf{y}-\mathbf{x}\|^2]=MN_0$. In the final form, the first two terms do not require the ground-truth signal $\mathbf{x}$. 
However, the last term has a dependency on $\mathbf{x}$. 
Fortunately, this term can be rewritten without explicit dependence on $\mathbf{x}$ by applying Lemma \ref{thm:steinslemma}.
\begin{lemma}[Stein's Lemma]\label{thm:steinslemma}
    Consider the signal $\mathbf{x}$ and its noisy observation $\mathbf{y}\sim\mathcal{CN}(\mathbf{x}, N_0\mathbf{I}_M)$. 
    Let $\hat{\mathbf{x}}(\mathbf{y})$ be a weakly differentiable denoiser function which estimates $\mathbf{x}$ from $\mathbf{y}$. Then, 
    \begin{align}
        & 2\mathbb{E}[(\hat{\mathbf{x}}(\mathbf{y}) - \mathbf{y})^H(\mathbf{x}-\mathbf{y})] \nonumber\\ 
        & = N_0 \mathbb{E} \left[ \left( \sum_{m=1}^M \frac{\partial\textnormal{Re}\{\hat{x}_m(y_m)\}}{\partial\textnormal{Re}\{y_m\}} + \frac{\partial\textnormal{Im}\{\hat{x}_m(y_m)\}}{\partial\textnormal{Im}\{y_m\}} \right) -2M \right].
    \end{align}
\end{lemma}
\begin{proof}
    Please refer to \cite{ref:steinslemma}.
\end{proof}
By substituting the result from Lemma \ref{thm:steinslemma} into the final term in \eqref{eqn:msederived}, we obtain the following expression,
\begin{align}
    &\frac{1}{M}\mathbb{E}[\|\hat{\mathbf{x}}(\mathbf{y}) - \mathbf{y}\|^2] + N_0 + \frac{2}{M}\mathbb{E}[(\hat{\mathbf{x}}(\mathbf{y}) - \mathbf{y})^H(\mathbf{x}-\mathbf{y})]\nonumber \\
    & = \mathbb{E}\bigg[\frac{1}{M}\|\hat{\mathbf{x}}(\mathbf{y})-\mathbf{y}\|^2+{N}_0 \nonumber\\ 
    & +\frac{{N}_0}{M}\sum_{m=1}^M \left( \frac{\partial\textnormal{Re}\{\hat{x}_m(y_m)\}}{\partial\textnormal{Re}\{y_m\}} + \frac{\partial\textnormal{Im}\{\hat{x}_m(y_m)\}}{\partial\textnormal{Im}\{y_m\}} -2 \right) \bigg] \nonumber \\
    & = \mathbb{E}[\mathcal{S}],
\end{align}
which completes the proof.\qed


\begin{thebibliography}{00}
\bibitem{ref:precoding} R. Jiao and L. Dai, ``On the Max-Min Fairness of Beamspace MIMO-NOMA,'' \textit{IEEE Trans. Signal Process.}, vol. 68, pp. 4919--4932, 2020.

\bibitem{ref:channelest} H. Park and J.-W. Choi, ``Binary Hypothesis Testing-Based Low-Complexity Beamspace Channel Estimation for mmWave Massive MIMO Systems,'' \textit{arXiv preprint} arXiv:2508.01007, 2025.

\bibitem{ref:equalization} Z. Hong, T. Li, C. Li, D. Wang and X. You, ``Group-Joint MMSE Complementary-Based Distributed Uplink for Cell-Free Massive MIMO,'' \textit{IEEE Trans. Wireless Commun.}, vol. 23, no. 10, pp. 13648--13663, 2024.

\bibitem{ref:scheduling} H. Park and J.-W. Choi, ``Queue-Aware Optimization-Based Scheduling for mmWave Multi-User MIMO Indoor Small Cell,'' \textit{IEEE Commun. Lett.}, vol. 29, no. 10, pp. 2303-2307, 2025.



\bibitem{ref:mmWavesurvey} S. A. Busari, K. M. S. Huq, S. Mumtaz, L. Dai and J. Rodriguez, ``Millimeter-Wave Massive MIMO Communication for Future Wireless Systems: A Survey,'' \textit{IEEE Commun. Surveys Tuts.}, vol. 20, no. 2, pp. 836--869, 2018.
\bibitem{ref:OFDM1} C. Shin, R. W. Heath, and E. J. Powers, ``Non-Redundant Precoding-Based Blind and Semi-Blind Channel Estimation for MIMO Block Transmission With a Cyclic Prefix,'' \textit{IEEE Trans. Signal Process.}, vol. 56, no. 6, pp. 2509--2523, 2008.
\bibitem{ref:OFDM2} J. Tian, T. Zhou, T. Xu, H. Hu and M. Li, ``Blind Estimation of Channel Order and SNR for OFDM Systems,'' \textit{IEEE Access}, vol. 6, pp. 12656--12664, 2018.
\bibitem{ref:OFDM3} T. Cui and C. Tellambura, ``Power Delay Profile and Noise Variance Estimation for OFDM,'' \textit{IEEE Commun. Lett.}, vol. 10, no. 1, pp. 25--27, 2006.
\bibitem{ref:EM1} A. Das and B. D. Rao, ``SNR and Noise Variance Estimation for MIMO Systems,'' \textit{IEEE Trans. Signal Process.}, vol. 60, no. 8, pp. 3929--3941, 2012.
\bibitem{ref:EM2} M. A. Boujelben, F. Bellili, S. Affes, and A. Stephenne, ``EM Algorithm for Non-Data-Aided SNR Estimation of Linearly-Modulated Signals over SIMO Channels,'' \textit{Proc. IEEE Global Telecommun. Conf. (GLOBECOM)}, 2009. 
\bibitem{ref:covariance1} J. P. González-Coma and D. Morales-Jiménez, ``Blind SINR Estimation for Massive MIMO Systems,'' \textit{IEEE Wireless Commun. Lett.}, vol. 13, no. 9, pp. 2492--2496, 2024.
\bibitem{ref:covariance2} S. K. Sharma, S. Chatzinotas and B. Ottersten, "Eigenvalue-Based Sensing and SNR Estimation for Cognitive Radio in Presence of Noise Correlation," \textit{IEEE Trans. Veh. Technol.}, vol. 62, no. 8, pp. 3671--3684, 2013.
\bibitem{ref:DL1} K. Tamura, S. Kojima, P. V. Trinh, S. Sugiura and C.-J. Ahn, ``Joint SNR and Rician K-Factor Estimation Using Multimodal Network Over Mobile Fading Channels,'' \textit{IEEE Trans. Mach. Learn. in Commun. Networking}, vol. 2, pp. 766--779, 2024.
\bibitem{ref:DL2} X. Xie, S. Peng and X. Yang, ``Deep Learning-Based Signal-to-Noise Ratio Estimation Using Constellation Diagrams,'' \textit{Mobile Inf. Syst.}, vol. 2020, pp. 1--9, 2020.
\bibitem{ref:DL3} S. Kojima, K. Maruta, Y. Feng, C.-J. Ahn, and V. Tarokh, ``CNN-Based Joint SNR and Doppler Shift Classification Using Spectrogram Images for Adaptive Modulation and Coding,'' \textit{IEEE Trans. Commun.}, vol. 69, no. 8, pp. 5152--5167, 2021.

\bibitem{ref:lowblind} A. Gallyas-Sanhueza and C. Studer, ``Low-Complexity Blind Parameter Estimation in Wireless Systems with Noisy Sparse Signals,'' \textit{IEEE Trans. Wireless Commun.}, vol. 22, no. 10, pp. 7055--7071, 2023.

\bibitem{ref:sure} D. L. Donoho and I. M. Johnstone, ``Adapting to Unknown Smoothness via Wavelet Shrinkage,'' \textit{J. Amer. Stat. Assoc.}, vol. 90, no. 432, pp. 1200--1224, 1995.

\bibitem{ref:poisson} M. R. Akdeniz et al., ``Millimeter Wave Channel Modeling and Cellular Capacity Evaluation,'' \textit{IEEE J. Sel. Areas Commun.}, vol. 32, no. 6, pp. 1164--1179, 2014.

\bibitem{ref:channel} D. Tse and P. Viswanath, \textit{Fundamentals of Wireless Communication}. New York, NY, USA: Cambridge University Press, 2005.

\bibitem{ref:lasso} R. Tibshirani, ``Robust Shrinkage and Selection via the Lasso,'' \textit{J. Roy. Stat. Soc. Ser. B, Methodol.}, vol. 58, no. 1, pp. 267--288, 1996.

\bibitem{ref:quadriga} S. Jaeckel, L. Raschkowski, K. Börner and L. Thiele, ``QuaDRiGa: A 3-D Multi-Cell Channel Model With Time Evolution for Enabling Virtual Field Trials,'' \textit{IEEE Trans. on Antennas Propag.}, vol. 62, no. 6, pp. 3242--3256, 2014.

\bibitem{ref:steinslemma} L. H. Y. Chen, L. Goldstein, and Q.-M. Shao, \textit{Normal Approximation by Stein's Method}. New York, NY, USA: Springer, 2011.
\end{thebibliography}
\end{document}